%% file: main.tex
\documentclass[acmsmall,nonacm]{acmart}

\usepackage{amsmath}
\usepackage{amsthm}
\usepackage{graphicx}
\usepackage{algorithm}
\usepackage[noend]{algpseudocode}
\usepackage{xcolor}
\usepackage{subcaption}
\usepackage{cleveref}
\usepackage{thmtools}
\usepackage{thm-restate}

\usepackage[cppcommentstyle,continuouslinenumbers,mainfont=small]{smartalgorithmic}

\input{pseudocode/definitions}

\title{DAGs for the Masses}

\author{Michael Anoprenko}
\affiliation{
  \institution{LTCI, T\'el\'ecom Paris, Institut Polytechnique de Paris}
  \country{France}
}
\email{manoprenko@gmail.com}

\author{Andrei Tonkikh}
\affiliation{
  \institution{Aptos Labs}
  \country{USA}
}
\email{andrei@aptoslabs.com}

\author{Alexander Spiegelman}
\affiliation{
  \institution{Aptos Labs}
  \country{USA}
}
\email{sasha@aptoslabs.com}

\author{Petr Kuznetsov}
\affiliation{
  \institution{LTCI, T\'el\'ecom Paris, Institut Polytechnique de Paris}
  \country{France}
}
\email{petr.kuznetsov@telecom-paris.fr}

\author{Anatoliy Zinovyev}
\affiliation{
  \institution{Boston University}
  \country{USA}
}
\email{tolik@bu.edu}

\author{Konstantin Shprenger}
\affiliation{
  \institution{Institut Polytechnique de Paris}
  \country{France}
}
\email{shprenger@ip-paris.fr}

\renewcommand{\shortauthors}{M. Anoprenko, A. Tonkikh, A. Spiegelman, P. Kuznetsov, A. Zinovyev, and K. Shprenger}

\begin{document}

\begin{abstract}
A recent approach to building consensus protocols on top of Directed Acyclic Graphs (DAGs) shows much promise due to its simplicity and stable throughput.
However, as each node in the DAG typically includes a \emph{linear} number of references to the nodes in the previous round, prior DAG protocols only scale up to a certain point when the overhead of maintaining the graph becomes the bottleneck.

To enable large-scale deployments of DAG-based protocols, we propose a \emph{sparse DAG} architecture, where each node includes only a \emph{constant} number of references to random nodes in the previous round.
We present a \emph{sparse} version of Bullshark---one of the most prominent DAG-based consensus protocols---and demonstrate its improved scalability.

Remarkably, unlike other protocols that use random sampling to reduce communication complexity, we manage to avoid sacrificing resilience: the protocol can tolerate up to $f<n/3$ Byzantine faults (where $n$ is the number of participants), same as its less scalable deterministic counterpart.
The proposed ``sparse'' methodology can be applied to any protocol that maintains disseminated system updates and causal relations between them in a graph-like structure.
Our simulations show that the considerable reduction of transmitted metadata in sparse DAGs results in more efficient network utilization and better scalability.

\end{abstract}

% \begin{CCSXML}
% <ccs2012>
%    <concept>
%        <concept_id>10002978.10003006.10003013</concept_id>
%        <concept_desc>Security and privacy~Distributed systems security</concept_desc>
%        <concept_significance>500</concept_significance>
%        </concept>
%  </ccs2012>
% \end{CCSXML}

% \ccsdesc[500]{Security and privacy~Distributed systems security}

% \keywords{Consensus, Atomic Broadcast, Byzantine Fault Tolerance, DAGs, Scalability, Sampling}

\maketitle

\newpage\setcounter{page}{1}

\section{Introduction}

One of the principal questions in modern computing is how to make sure that a large number of users can share data in a consistent, available, and efficient way.
In blockchain systems~\cite{bitcoin,ethereum}, the users solve this problem by repeatedly reaching \emph{consensus} on an ever-growing sequence of data blocks (the problem formally stated as \emph{Byzantine atomic broadcast}, BAB).
Consensus is a notoriously difficult task~\cite{flp,dls,pbft,CHT96} that combines disseminating data blocks among the users and making sure that the blocks are placed by all users in the same order.
A popular trend of DAG-based consensus (\cite{aleph,dagrider,bullshark,shoal,shoalplusplus,mysticeti-tr,sui}, to name a few) is to separate data dissemination from the ordering logic.
In these protocols, all users disseminate blocks, and each block references a number of \emph{parents}---previously disseminated blocks---thus forming a directed acyclic graph (DAG).
The key observation is that by locally interpreting snapshots of this ever-growing graph, the users can infer a consistent ordering of the blocks.
Compared to more traditional, leader-based solutions such as PBFT~\cite{pbft}, this approach is simpler and exhibits a more stable throughput.
Moreover, as all users share the load of data dissemination, it naturally scales to a larger number of participants.

The original DAG-based BAB protocols are, however, very heavy on the communication channels, which inherently limits their scalability.
In this paper, we take a \emph{probabilistic} approach to build a slimmer DAG structure, where each block carries \emph{quasi-constant} amount of data about its predecessors, which results in a much less bandwidth-demanding solution.
We show that the reduced complexity comes with the cost of a very small probability of consistency violation and a slight latency increase.

\paragraph*{DAG consensus}

DAG-based consensus protocols proceed in (asynchronous) rounds.
In each round, each \emph{validator} (a user involved in maintaining the DAG) creates a node (typically containing a block of transactions submitted by the blockchain users)  and disseminates it using reliable broadcast~\cite{bra87asynchronous,cachin2011introduction}.
The node contains references to $n - f$ nodes from the previous round, where $n$ is the total number of validators and $f$ is the maximum number of faulty validators (typically, $f = \lfloor \frac{n - 1}{3} \rfloor$).
In every round, a dedicated node is chosen as an \emph{anchor}, where the choice can be either made deterministically based on the round number or at random using a common coin.
The goal of consensus is now to place the anchors in an ordered sequence.
In creating the sequence, the DAG nodes referencing the anchor as their parent are interpreted as its \emph{votes}.
Once the sequence of anchors is built, all DAG nodes in an anchor's causal past are placed before it in some deterministic way.

\paragraph*{The burden of metadata}

\begin{table}
\begin{center}
\begin{tabular}{|c|c|c|c|c|}
\hline
& \multicolumn{2}{c|}{Bullshark} & \multicolumn{2}{c|}{Sparse Bullshark} \\
\cline{2-5}
& complexity & ex. egress traffic & complexity & ex. egress traffic \\
& &  & &  \\
\hline
Threshold signatures & $\Theta(n^2 \lambda)$ & 171 MB & $\Theta(n \lambda^2)$ & 17 MB \\
Multi-signatures     & $\Theta(n^3)$         & 837 MB & $\Theta(n^2 \lambda)$     & 81 MB \\
Regular signatures   & $\Theta(n^3 \lambda)$ & 171 GB & $\Theta(n^2 \lambda^2)$   & 16 GB \\
\hline
\end{tabular}
\end{center}
\caption{Communication complexity of Bullshark and Sparse Bullshark for different cryptographic setups and the estimated amount of egress traffic per-user per-round in the case when $n=2\,000$ and $\lambda=128$.}
\label{tab:comparison}
\end{table}

The central claim for DAG-based protocols is that they are able to amortize the space taken by the graph edges (\emph{metadata}) provided enough user transactions (\emph{payload}): If each block includes $\Omega(n)$ user transactions, the metadata will take at most a constant fraction of the block space.
However, while alluring in theory, this claim does not stand the test of reality where we cannot simply keep increasing the block size indefinitely due to limited resources (network, storage, and CPU) and a limited demand.

As each user needs to receive, process, and store each other user's block in each round, and each block includes $n-f$ references, per-user per-round communication, computational, and storage complexity is at least $\Omega(n^2)$ multiplied by the size of a single reference.
In practice, in most protocols~\cite{dagrider,narwhal,bullshark,shoal,shoalplusplus}, each reference includes either $n-f$ signatures, a multi-signature~\cite{multisignatures,bls-multisignatures} with a bit-vector of size $n$, or a threshold signature~\cite{threshold-cryptosystems,bls-threshold,das2023threshold}, which brings the per-user per-round complexity to $\Omega(n^3 \lambda)$, $\Omega(n^3)$, or $\Omega(n^2 \lambda)$ respectively, depending on the cryptographic primitives available (see \Cref{tab:comparison}), where $\lambda$ is the security parameter.\footnote{Intuitively, one can think of $\lambda$ as the number of ``bits of security'', i.e., as the logarithm of the number of operations required to break the cryptosystem. Most hash functions and digital signature schemes have output of size $O(\lambda)$. However, multi-signatures require additional $O(n)$ bits to store the set of nodes who contributed to the signature. In this paper, we assume $n \gg \lambda$.}

This issue is a significant barrier not only for effective resource utilization, but also for the scalability of the system: modern production blockchain deployments operate with around $100$ validators.

\paragraph*{Our contribution}

We introduce a technique for reducing the number of edges in the DAG by randomly sampling parents of each DAG vertex from a quorum of vertices in the previous round, leaving only $O(\lambda)$ outgoing edges from each vertex.
The communication complexity is thus brought down to $O(n^2 \lambda^2)$ with plain signatures, $O(n^2 \lambda)$ with multi-signatures, or $O(n \lambda^2)$ with threshold signatures.
To anticipate the adversary populating a sample with Byzantine nodes only, we use \emph{verifiable sampling}. 
We equip the sample of a vertex with a \emph{proof} that it has indeed been fairly selected from a sufficiently large ($\geq 2f+1$) set of the parent vertices.
We show that the probability that a computationally bounded adversary is able to fabricate such a proof for a sample of its choice is negligible.

Sparse sampling significantly lowers the overhead of DAG-based protocols, enabling larger-scale deployments.
We present \emph{Sparse Bullshark}, a variant of Bullshark~\cite{bullshark}---a state-of-the-art DAG-based consensus protocol, that exhibits reduced communication complexity and achieves better scalability.
In \Cref{tab:comparison}, we evaluate the communication complexity and egress per-node per-round traffic of Sparse Bullshark, compared with the original Bullshark protocol.

Unlike prior solutions to scaling distributed systems with random sampling (\cite{gilad2017algorand,coinicidence,abraham2025asynchronous,blum2020asynchronous}, to name a few), our approach maintains the optimal resilience and tolerates up to $f<n/3$ Byzantine faults.

Our simulations show that Sparse Bullshark allows for deployment with thousands of validators with realistic network bandwidth assumptions.

\paragraph*{Roadmap}
In \Cref{sec:related}, we overview related work.
In~\Cref{sec:model}, we briefly describe our system model and recall the BAB specification.
In~\Cref{sec:protocol}, we present our protocol, highlight its key differences with earlier work.
In~\Cref{sec:analysis}, we argue its correctness and review its communication complexity.
In~\Cref{sec:eval}, we present the outcomes of our simulations.
\Cref{sec:conc} concludes the paper.
Proofs of correctness are delegated to the appendix. 

\section{Related work}
\label{sec:related}

Byzantine atomic broadcast (BAB), explicitly defined by Cachin et al.~\cite{bab} is equivalent to Byzantine fault-tolerant state-machine replication (SMR)~\cite{pbft}, extending original crash fault-tolerant abstractions~\cite{HT93,paxos} to the Byzantine setting.
The advent of blockchain~\cite{bitcoin,ethereum}, by far the most popular BAB application, brought a plethora of fascinating algorithms, optimized for various metrics, models, and execution scenarios (\cite{hotstuff,hyperledger,gilad2017algorand}, to name just a few).

Any BAB protocol implements two basic functionalities: \emph{data dissemination}, ensuring that every transaction (data unit) reaches every correct client, and \emph{ordering}, ensuring that the disseminated transactions are processed in the same order.
To simplify the protocol design and ensure stable throughput despite contention and failures, a class of BAB protocols~\cite{aleph,dagrider,narwhal,bullshark} delegate all communication between the system participants to data dissemination.
The (blocks of) transactions are reliably broadcast, the delivered block and their causal relations form a directed acyclic graph (DAG).

DAG-based consensus protocols based on reliable broadcast~\cite{dagrider, cordialminers, narwhal, bullshark, shoal, shoalplusplus} differ in the choice of how to select the anchors and how to decide which of them to commit.
The ordering logic in DAG-based protocols requires no additional communication, and can be implemented assuming common coins~\cite{dagrider} or partial synchrony~\cite{bullshark-ps}.
Maintaining \emph{certificates}, each containing $2f + 1$ signed references, requires significant bandwidth for sending them around and incurs an intensive load on the CPU for the generation and verification of signatures.

More recently, several proposals~\cite{cordialminers, mysticeti-tr,mahimahi} considered \emph{uncertified} DAGs with weaker consistency properties, that can be implemented on top of weaker broadcast primitives.
The resulting protocols intend to improve the latency of original solutions and lighten the CPU load, but may be more susceptible to network attacks as demonstrated in~\cite{shoalplusplus}.

In this paper, we chose the partially synchronous version of Bullshark~\cite{bullshark,bullshark-ps} as the basis for the discussion because it is sufficiently established in the field and because its relative simplicity allows us to focus on the core idea of a sparse DAG, not being distracted by a range of optimizations employed in later protocols~\cite{shoal,shoalplusplus,mysticeti-tr,mahimahi}.

Multiple BAB or Byzantine agreement protocols (e.g., \cite{gilad2017algorand,coinicidence,abraham2025asynchronous,blum2020asynchronous})
rely on randomly elected committees to improve scalability, assuming a weaker system resilience ($f<(n-\varepsilon)/3$ for some constant $\varepsilon$).
In this paper, however, we explore an alternative route to scaling through randomization, not based on committee election.
One key advantage of this new approach is that it maintains the original, theoretically optimal, resilience of $f<n/3$.

To alleviate the inherent communication complexity of classical \emph{quorum-based} reliable broadcast algorithms~\cite{bra87asynchronous,MR97srm,cachin2011introduction}, numerous \emph{probabilistic} broadcast algorithms resorted to \emph{gossip}-based communication patterns~\cite{Berenbrink2008,Elssser2015OnTI,Alistarh2010,Fernandess2007,Sourav2018SlowLF,Giakkoupis2016}.
In particular, Scalable Byzantine reliable broadcast~\cite{scalable-rb-disc} replaces quorums with random $O(\log n)$ samples, which allows it to achieve, with high probability, significantly lower communication complexity with a slight increase in latency.

\section{Preliminaries}

\label{sec:model}
We use the standard BFT system model.
For simplicity, we assume $n = 3f + 1$ participants $\Pi = \{p_1, \ldots, p_n\}$ (or \emph{validators}), up to $f$ of which might be \emph{Byzantine} faulty.
A faulty process may arbitrarily deviate from the protocol.
We assume partial synchrony: there is a known upper bound $\Delta$ on the network delays that, however, holds only after an unknown (but finite) time (called the \emph{global stabilization} time, GST).
We assume that all participants and the adversary are computationally bounded, i.e., they can perform at most $poly(\lambda)$ operations in the duration of the protocol, where $\lambda$ is the security parameter used to instantiate the cryptographic primitives.
We allow for an \emph{adaptive} adversary, i.e., it can select which nodes to corrupt (subject to the adaptive security of the digital signature scheme used).

\paragraph*{Reliable broadcast}
As a building block we use \textit{Byzantine Reliable Broadcast} primitive.
It is called via $r\_bcast(m, r)$ and produces events $r\_deliver(m, r, p_k)$, where $m$ is the message being broadcast, $p_k$ is the sender, and and $r$ is a sequence number (local and monotonically growing for each sender).
Reliable broadcast guarantees following properties:
\begin{itemize}
  \item \textbf{Validity}: if a correct validator $p_k$ invokes $r\_bcast(m, r)$ then eventually $p_k$ outputs $r\_deliver(m, r, p_k)$;
  \item \textbf{Agreement}: if a correct validator outputs $r\_deliver(m, r, p_k)$ then eventually every correct validator outputs $r\_deliver(m, r, p_k)$;
  \item \textbf{Integrity}: for each $r \in \mathbb{N}$ and $p_k \in \Pi$ every correct process outputs $r\_deliver(m, r, p_k)$ at most once (for any $m$);
\end{itemize}

\paragraph{Vector commitment}
Our protocol uses a \textit{vector commitment} (VC) scheme exporting the functions:
\begin{itemize}
  \item $VC.Commit([a_0, \ldots, a_{n - 1}]) \mapsto commitment$: generate proof of the commitment;
  \item $VC.Proof([a_0, \ldots, a_{n - 1}], i) \mapsto proof$: generate a proof that element $i$ equals to $a_i$;
  \item $VC.Verify(commitment, i, a_i, proof) \mapsto \{true, false\}$: verify that generated proof corresponds to the commitment.
\end{itemize}
In particular, we use the KZG vector commitment scheme~\cite{kzg} that provides commitments and proofs of constant size.

\paragraph{Approximate lower bound arguments}
For sampling in our protocol we use \textit{approximate lower bound arguments} (ALBA)~\cite{alba}.
In particular, we use the Telescope ALBA abstraction with parameters $n_p$, $n_f$ and $\SampleSize$.
The abstraction generates a sample of size $\SampleSize$ from a set of at least $n_p$ elements (satisfying a given predicate $A$), equipped with a succinct "proofs of possession" that the sample is indeed selected from a sufficiently large set satisfying $A$.
The abstraction exports two operations:
\begin{itemize}
  \item $ALBA.Prove(seed, S_p) \mapsto sample$: generate a sample of size $\SampleSize$ from elements of a set $S_p$ where $|S_p| \geq n_p$ with a given seed;
  \item $ALBA.Verify(seed, sample) \mapsto \{true, false\}$: verify the generated sample with a given seed.
\end{itemize}
It is ensured that for any sample constructed from a set of size at least $n_p$, verification succeeds with high probability.
At the same time, if the sample was constructed (by a computationally bounded adversary) from a set of at most $n_f$ elements, then verification succeeds only with a negligible probability.
In particular, by \cite[Corollary~1]{alba}:
\begin{itemize}
\item an honest party will successfully construct a valid proof having a set of at least $n_p$ valid elements with high probability;
\item computationally bounded adversary will not be able to construct a valid proof having a set of at most $n_f$ elements with high probability.
\end{itemize}
Telescope ALBA proofs consist of only $\Theta(\SampleSize)$ sampled elements provided that the ratio of $n_p / n_f$ is bounded by a constant.

We use a variation of the Telescope ALBA where prover and verifier are additionally provided with a randomness seed.
Seeding is used in our protocol to make sampling results unpredictable.
It can be easily implemented by ``salting'' all the hash values used in ALBA prover and verifier with the provided seed value.
We refer to~\cite{alba} for details on its implementation.

\paragraph*{Problem statement}
We implement the \emph{Byzantine Atomic Broadcast} (BAB) abstraction.
It is called via $a\_bcast(m, r)$ and produces events $a\_deliver(m, r, p_k)$, where $m$ is the message being broadcast, $p_k$ is the sender, and $r$ is a sequence number (local for the sender).
The abstraction satisfies the same properties as the Byzantine Reliable Broadcast, with additional property of \textbf{Total order}: if a correct process outputs $a\_deliver(m, r, p_k)$ before $a\_deliver(m', r', p'_k)$ then no correct process outputs $a\_deliver(m', r', p'_k)$ before $a\_deliver(m, r, p_k)$.
We require all these properties to be satisfied \emph{with high probability}, i.e., the error probability is exponentially small in the security parameter $\lambda$.

\section{Protocol}

\label{sec:protocol}
We begin by recalling the basic operation of the partially synchronous version of Bullshark~\cite{bullshark,bullshark-ps}, as we make use of its DAG construction and commit rules.
Then we describe Sparse Bullshark, with the focus on its key differences from prior work.
In particular, we present our verifiable sampling mechanism that anticipates the adversarial attempts to negatively affect the DAG connectivity.

\subsection{Partially synchronous Bullshark}

\input{pseudocode/common}
\input{pseudocode/es_bullshark}

Bullshark, as well as other DAG-based consensus protocols, has two conceptual elements: DAG construction and ordering.
All communication happens in the DAG construction phase, while ordering is a purely local procedure: a validator commits new blocks by interpreting its local version of the DAG.

DAG construction is divided into communication \emph{rounds}.
In every round $r$, each validator creates a new DAG node and shares it with other validators using reliable broadcast~\cite{bra87asynchronous,cachin2011introduction}.
The round-$r$ node contains a batch of \emph{transactions} (application-specific operations) as well as $n - f$ or more references to round-$(r - 1)$ nodes.
Validator can proceed to round $r + 1$ as soon as it receives at least $n - f$ $r$-round nodes via reliable broadcast.

Along with the construction of the DAG, each validator is constantly trying to order some nodes of its current copy of the DAG.
Rounds are assembled in \emph{waves}, each wave consists of two rounds.
Every even round contains the \emph{anchor} of the wave---a predefined ``leader'' node that validators will try to commit in the ordering.
Various strategies can be employed to select anchors: deterministically in a round-robin fashion, randomly by using a common coin primitive or dynamically using the leader reputation mechanism.
To determine whether an anchor should be committed, several \emph{commit rules} may be applied:
\begin{enumerate}
  \item \textit{Direct commit rule}:
  if $f + 1$ nodes in round $r + 1$ are connected to the anchor of round $r$, it is committed;
  \item \textit{Indirect commit rule}:
  recursively, before committing an anchor node, every prior anchor node in its causal past is committed first.
\end{enumerate}

Therefore, anchors are committed in the order of their rounds.
Additionally, whenever an anchor is committed, all (not yet ordered) nodes in its causal history are placed in the global ordering (and the corresponding messages are delivered to the application), using some deterministic rule.

Safety of Bullshark relies on the following invariant:
if a round-$r$ anchor $A$ is committed at a correct validator using the direct commit rule, then the anchor of every future round will have a path to it.
Thus, if $A$ has been committed by a correct validator using a direct commit rule, then it will be committed by all other correct validators using direct or indirect commit rules.
This way, the correct validators agree on the  order in which anchors are committed, and hence on the ordering of the whole DAG.

The pseudocode of eventually synchronous Bullshark is presented in four parts.
Algorithm~\ref{util} presents a function that checks if a path between two  DAG node exists and a function that determines the anchor of a  given round.
Algorithm~\ref{construction} describes the DAG construction, and Algorithm~\ref{ordering}---the ordering procedure.
In Algorithm~\ref{esbs}, we describe the DAG node structure, the procedure for creating a new node, the procedure for checking the validity of received node, and a function that determines how many votes are necessary to apply the direct commit rule to an anchor (compared to $f + 1$ votes used in Bullshark,  we use $2f+1$ votes, as we show below).

\subsection{Sparse Bullshark} \label{subsec:sparse}

\input{pseudocode/sparse_bs}

We obtain Sparse Bullshark from the baseline Bullshark with a few modifications presented in \Cref{alg:new}.
The functions in \Cref{alg:new} replace those in \Cref{esbs}, the modifications we introduce are highlighted in blue.

\subsubsection{Verifiable sampling}

The key distinguishing feature of Sparse Bullshark is that a DAG vertex stores only a \emph{sample} of vertices of size $\Theta(\lambda)$ from the previous round instead of a full quorum.
We denote the sample size by $\SampleSize$.
Intuitively, when a correct validator collects $n - f$ round-$r$ vertices, it selects, uniformly at random, a subset of size $\SampleSize$ of them and creates a new vertex in the round $r + 1$ referencing the selected subset.
The specific choice of $\SampleSize$ affects communication complexity, latency, and the security of the protocol.
More specifically, as we show in~\Cref{sec:correctness} the protocol is secure against computationally bounded adversaries as long as $\SampleSize$ is $\Omega(\lambda)$.
We discuss the impact of the sample size on latency in detail in \Cref{subsec:latency}. 

Instead of taking a random sample, Byzantine validators might attempt to undermine the safety and liveness of the protocol by maliciously selecting a set of $\SampleSize$ vertices from the previous round in order to violate the protocol's safety property.
To anticipate this behavior, we employ \textit{verifiable sampling}:
Each validator must provide a proof that it has indeed collected a quorum of $n - f$ nodes in a round and has sampled $\SampleSize$ of them in a pseudo-random manner.

When creating a vertex in round $r$, validator $p_i$ constructs this proof by:
\begin{enumerate}
\item
committing to all the certificates of the collected quorum of vertices from round $r - 1$ using vector commitment.
The commitment is performed for $\dagvar_i[r - 1]$ interpreted as an array: for each sender $j$ we take the certificate of a vertex created by $p_j$ or $\bot$ if no such vertex is contained in the local DAG view of $p_i$.
This step prevents the adversary from manipulating with edges of vertices created by the Byzantine parties after the sampling is performed.
\item
constructing a sample of size $\SampleSize$ from certificates of the collected quorum using Telescope $ALBA.Prove$ with parameters $n_p = 2f + 1$ and $n_f = f$.
With high probability a valid sample will be constructed.
Edges are then created to all the vertices from the sample.
When broadcasting the new vertex, the validator attaches the initial commitment, opening proofs for all the sampled elements, and the corresponding ALBA proof.
\end{enumerate}

Upon receiving a vertex from reliable broadcast, every validator verifies both vector commitment proof and ALBA proof attached to the vertex.
Any vertices with invalid proofs are disregarded.

This approach, while restricting the behavior of Byzantine nodes to some degree, still leaves open the possibility of \emph{grinding attacks}: a malicious validator that receives more than $n - f$ round-$r$ blocks may try many different subsets of size $n - f$ to find a more ``favorable'' sample.
We show below, however, that as long as $\SampleSize = \Omega(\lambda)$, this attack is, intuitively, as efficient as any other naive attack on the cryptographic primitives, such as simply trying to recover a correct node's private key by a brute force search.

\subsubsection{Referencing the anchor}

Since the references to the anchor in round $r$ by the nodes in round $r + 1$ are used as votes for the direct commit rule, on top of selecting $\SampleSize$ random parents, if a node finds the anchor of the previous round in its quorum, it includes the anchor into its sample (\cref{line:include-anchor}).
Notice that trying to enforce this rule on Byzantine users would be unhelpful as they can always simply pretend to not have received the anchor node.

\subsubsection{Direct commit rule}
\label{subsec:direct-commit}

Finally, we modify the direct commit rule, raising the threshold from $f + 1$ to $n - f = 2f + 1$.
This doesn't affect the protocol's correctness.
Indeed, if the user producing the anchor node is correct and the network is in the stable period, all (at least $n - f$) correct users will reference the anchor node in round $r + 1$.
Intuitively, the modification ensures that, with high probability, every committed anchor is reachable (through a sequence of samples) from every consequent anchor, which guarantees \textbf{Total order}.

\section{Correctness and complexity analysis}
\label{sec:analysis}

\subsection{Correctness}
\label{sec:correctness}

As with prior DAG-based protocols, by using reliable broadcast,
the honest users reach a form of eventual agreement on their DAGs: if a node appears in the DAG of a correct user, every other correct user will eventually have the node in its DAG.
Moreover, the honest users agree on the causal history of that node.

The \textbf{Integrity} property follows from Integrity of reliable broadcast and the DAG construction rules: a broadcast message can be part of at most one DAG node; and as a message is delivered only when the corresponding DAG node appears in the causal history of some anchor for the first time, the message is delivered at most once.

Intuitively, to ensure \textbf{Total order}, it is necessary and sufficient to maintain the key invariant of Bullshark: if an anchor is directly committed, then it will be in the causal history of any future anchor.

\begin{restatable}{theorem}{ThTO}
\label{th:total-order}
Sparse Bullshark implements the \textbf{Total order} property of BAB with high probability.
\end{restatable}

The properties of of \textbf{Validity} and \textbf{Agreement} follow, intuitively, from the fact that after the GST all the anchors will be referenced by all the correct validators and thus committed.
Therefore, every vertex will eventually be included in a causal history of some committed anchor and all the messages will be delivered.

\begin{restatable}{theorem}{ThVA}
\label{th:agreement-validity}
Sparse Bullshark implements the \textbf{Agreement} and \textbf{Validity} properties with high probability.
\end{restatable}

The proofs are delegated to Appendix~\ref{app:proofs}.

\subsection{Latency}
\label{subsec:latency}

During periods of synchrony, all blocks created by correct validators in odd rounds will link to the anchor of the previous round because of the additional timer delay before creating a new vertex.
Therefore, all anchor nodes by correct validators will be committed.
The frequency of anchors being committed in Sparse Bullshark is the same as in the original Bullshark.

However, as the anchor node in Sparse Bullshark only references $\SampleSize$ nodes from the previous round instead of $n - f$, the protocol incurs additional \emph{inclusion latency}: the delay between a vertex being broadcast and being included in some anchor's causal history.
Assuming that each node has its parents selected at random from the previous round\footnote{For both baseline Bullshark and Sparse Bullshark, to analyze the inclusion latency, we assume that each user receives nodes from a random quorum in each round. Indeed, otherwise up to $f$ users may always be late and have very large inclusion latency (until they become anchors themselves).}, the exact delay will follow the rumor-spreading pattern~\cite{rumor-spreading,rumor-spreading-multicall,rumor-spreading-revisited}.

\begin{figure}
    \centering
    \begin{subfigure}{.5\linewidth}
        \centering
        \includegraphics[width=\linewidth]{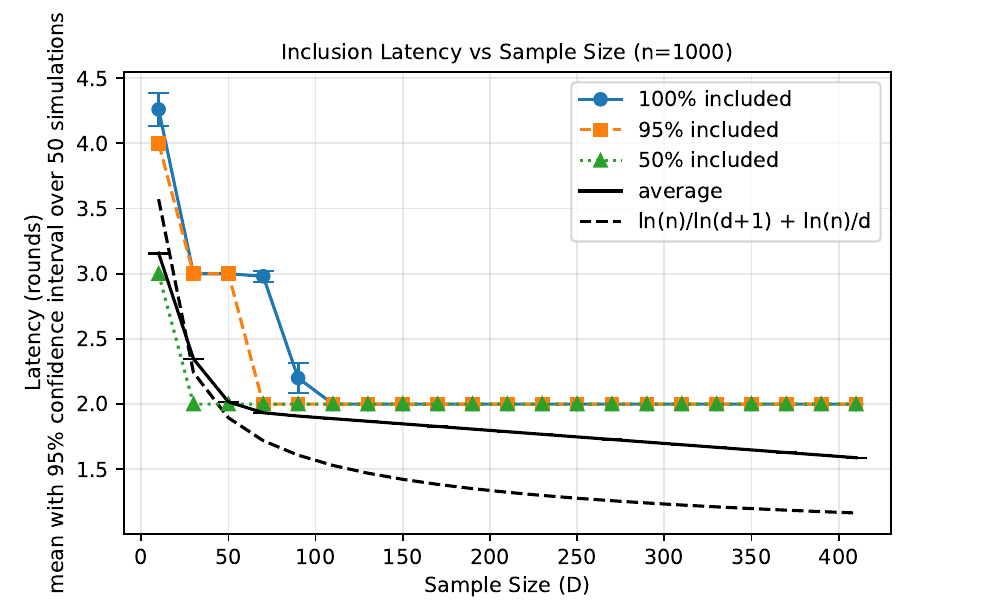}
        \caption{inclusion latency for $n=1000$ nodes}
        \label{subfig:inc-latency-1k}
    \end{subfigure}%
    \begin{subfigure}{.5\linewidth}
        \centering
        \includegraphics[width=\linewidth]{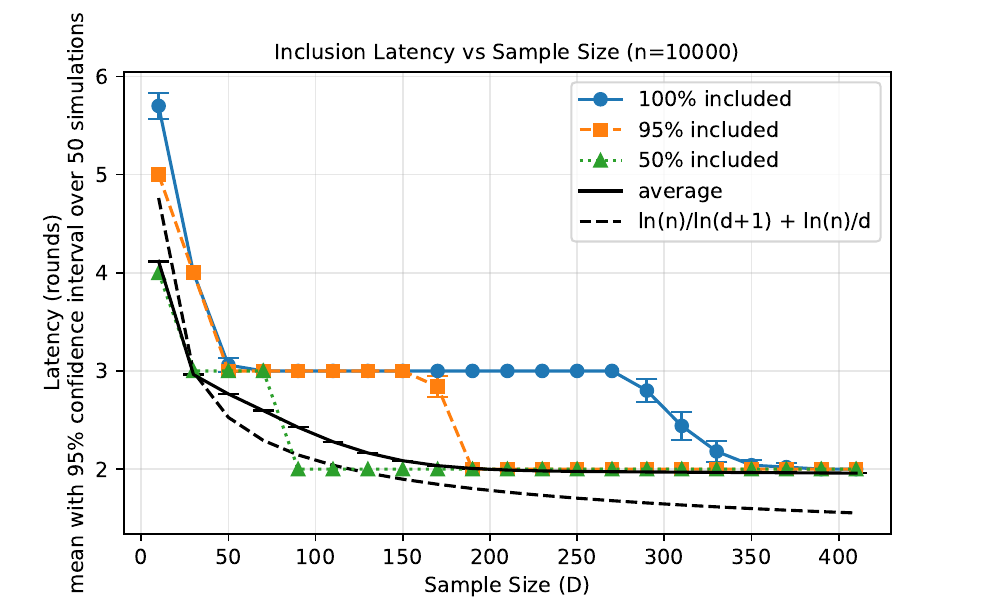}
        \caption{inclusion latency for $n=10\,000$ nodes}
        \label{subfig:inc-latency-10k}
    \end{subfigure}

    \caption{Inclusion latency (in DAG rounds) simulation, for varying sample size ($\SampleSize$), assuming each node has $\SampleSize$ nodes from the previous round uniformly at random as its parents.}
    \label{fig:inc-latency}
\end{figure}

The inclusion delays (in terms of DAG rounds) measured in our simulations are presented in \Cref{subfig:inc-latency-1k,subfig:inc-latency-10k}, for 1000 and 10\,000 users, respectively.
For example, selecting $\SampleSize = 70$ for 1000 users or $\SampleSize = 190$ for 10\,000 users is enough to ensure that 95\% of nodes will have inclusion latency of at most 2, i.e., will be referenced by some node referenced directly by the anchor.
In contrast, Bullshark has inclusion latency of $1$ for the $n - 1$ nodes referenced directly by the leader and $2$ (assuming random message ordering, with high probability) for the rest of the nodes.

\subsection{Communication complexity} \label{subsec:complexity}

We estimate the per-node metadata communication complexity.
In Sparse Bullshark, each DAG node contains up to $\SampleSize + 2 = \Theta(\lambda)$ references to other nodes compared to $\Theta(n)$ references in Bullshark.
Each reference is a certificate for the corresponding DAG node and its size is:
\begin{itemize}
\item $\Theta(\lambda)$ in case threshold signatures are available;
\item $\Theta(n + \lambda) = \Theta(n)$ in case multi-signatures are available;
\item $\Theta(n \lambda)$ otherwise.
\end{itemize}
We will denote the size of a certificate by $C$.

In Sparse Bullshark, the metadata of a node consists of $\Theta(\lambda)$ references, so metadata size for each node is $\Theta(\lambda C)$.
In contrast, in Bullshark, the metadata consists of $\Theta(n)$ references, so the size of a node is $\Theta(nC)$.
Additionally, each vertex contains a vector commitment proof that has constant size with KZG scheme, as well as an ALBA proof of size $O(\lambda)$.

Practical implementations rely on a combination of Signed Echo Consistent Broadcast~\cite{cachin2011introduction} with randomized \emph{pulling}\cite[Section 4.1]{narwhal}.
This approach allows us to surpass the theoretical bound on communication complexity of deterministic reliable broadcast~\cite{dolev-reischuk-bound} and achieve linear expected communication complexity.
With this approach, the expected communication complexity is $\Theta(n^2 C)$ for Bullshark and $\Theta(n \lambda C)$ for Sparse Bullshark.
In \Cref{tab:comparison}, we present the resulting communication complexity and estimated user egress traffic that our protocol gives for different cryptographic setups. 
Compared to Bullshark, we win a factor of $n$ in asymptotic communication complexity and an order of magnitude in the estimated generated traffic.

\section{Evaluation}
\label{sec:eval}

We evaluated our Sparse Bullshark protocol using a high-performance discrete event network simulator implemented in Rust.\footnote{\url{https://anonymous.4open.science/r/sparse_bullshark_evaluation-3AB5/}}
In our evaluations for emulated message latencies, we used ($50$, $10$)-normal distribution for $99\%$ of messages and ($500$, $10$)-normal distribution for a randomly selected 1\% of the messages.
We limited the network bandwidth of each simulated validator in order to emulate a scenario where the amount of transmitted data gets close to the network capacity on a host and starts affecting the latency and throughput of the protocol.

Our goal was to evaluate the resources consumed by our protocol on the \emph{metadata} exchange, compared to Bullshark~\cite{bullshark-ps}.
Evaluations were performed for the network of $1000$ validators.

\begin{figure}
    \centering
    \begin{subfigure}{.5\linewidth}
        \centering
        \includegraphics[width=\linewidth]{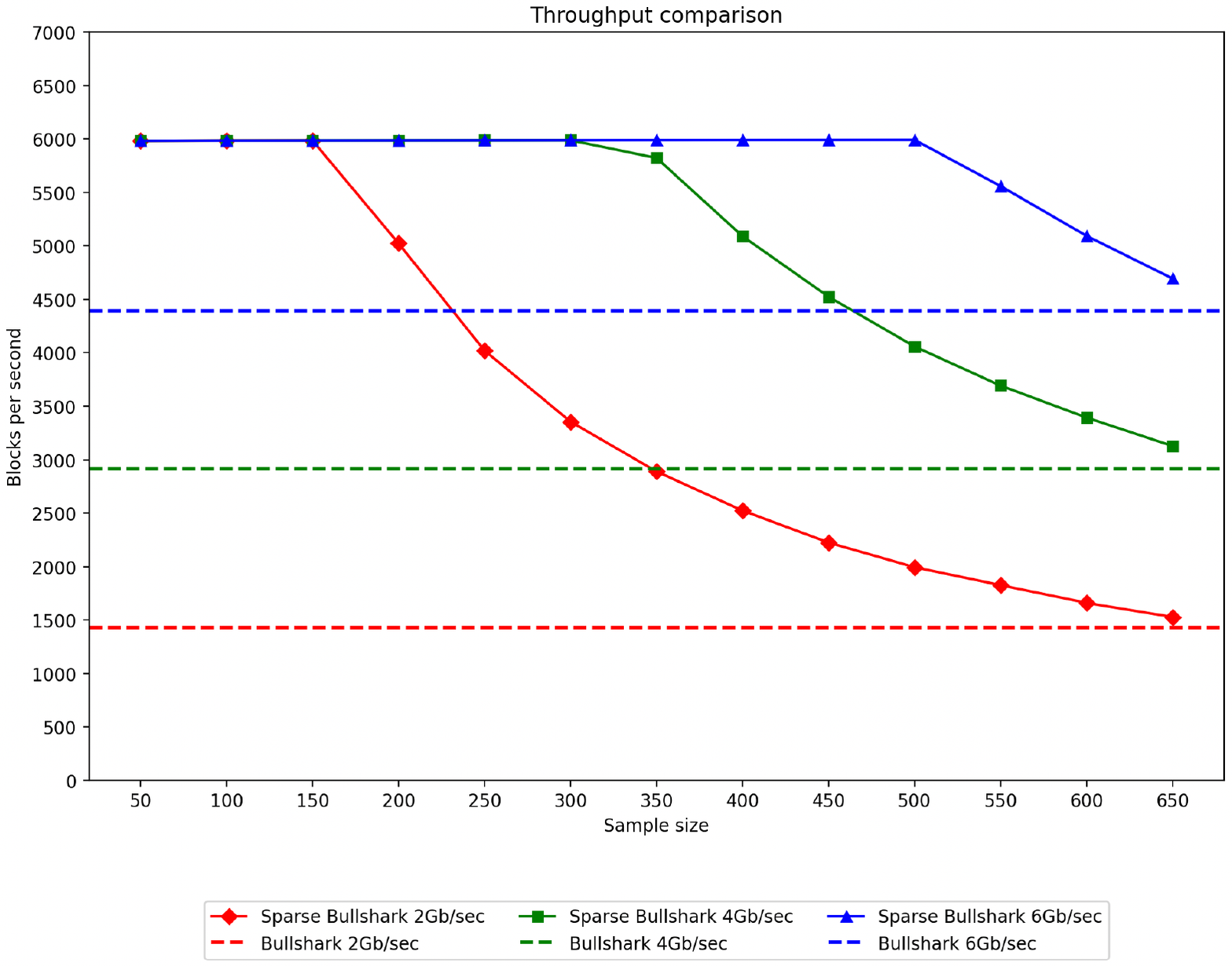}
        \caption{Throughput}
        \label{fig:tput}
    \end{subfigure}%
    \begin{subfigure}{.5\linewidth}
        \centering
        \includegraphics[width=\linewidth]{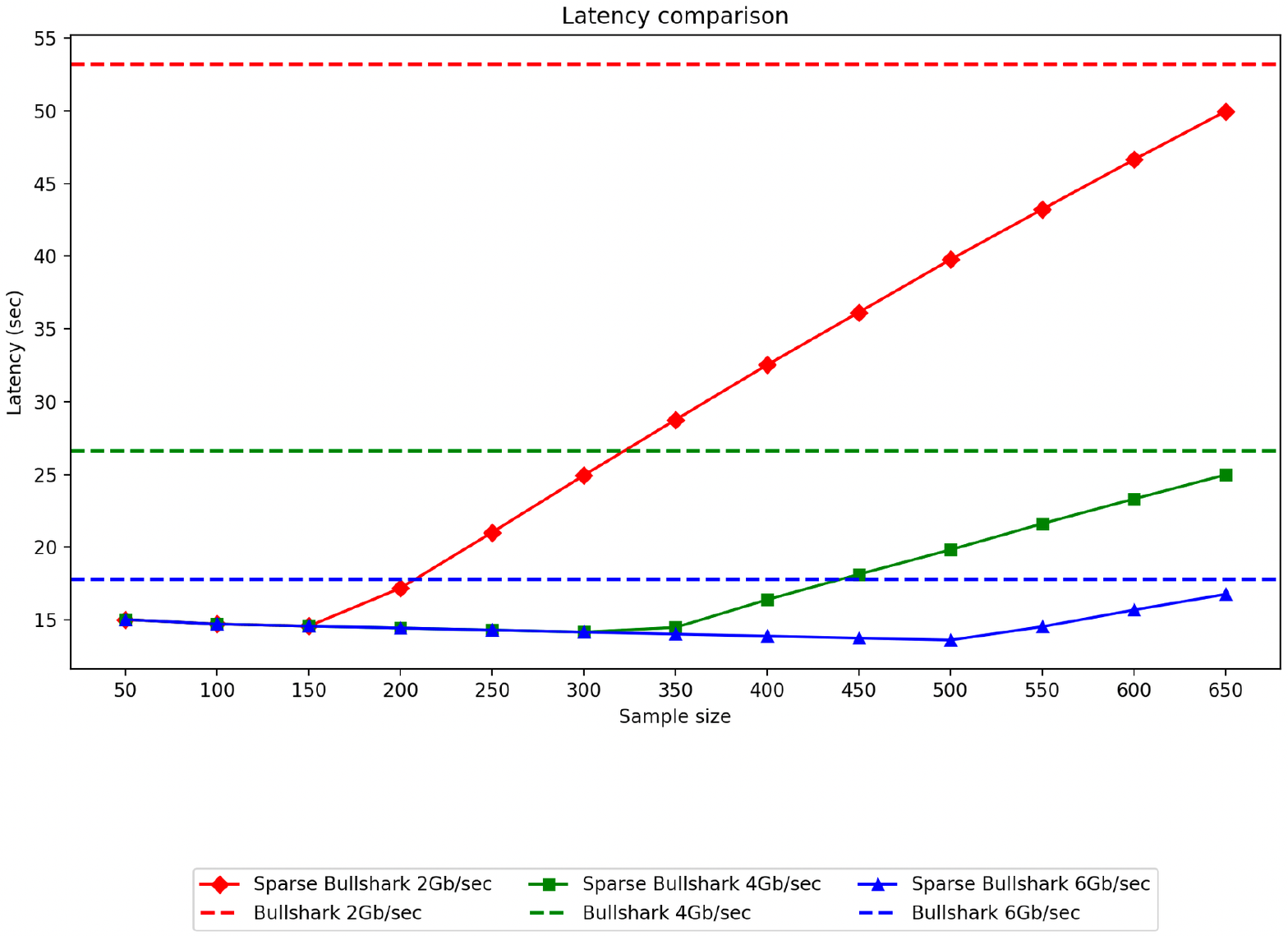}
        \caption{Average commit latency}
        \label{fig:latency}
    \end{subfigure}

    \caption{Simulated statistics of Bullshark and Sparse Bullshark with various bandwidth limits and sample sizes}
\end{figure}

Figure~\ref{fig:tput} relates the throughput of the protocol (in blocks committed per second) to the sample size, for various per-node network bandwidth limits.
We can observe that small sample sizes dramatically increase the throughput, allowing the protocol to be operate normally when, under the same conditions, Bullshark saturates the network.

Figure~\ref{fig:latency} shows the change in latency depending on the sample size for various per-node network bandwidth limits.
Despite the theoretical result of increased latency in terms of message delays in the worst case, we can see that in practice due to more efficient network utilization, Sparse Bullshark achieves lower latency under limited network bandwidth.

Evaluations confirm that while the sampling significantly decreases the amount of metadata transmitted through the network, which also leads to a decrease in latency when the bandwidth restriction is below a certain threshold.

\section{Concluding remarks}
\label{sec:conc}

In this paper, we propose a novel architecture to achieve better scalability in DAG-based consensus protocols: We construct a \emph{sparse DAG}, where each node references only a small random subset of nodes in the previous round.
To illustrate this approach, we present \emph{Sparse Bullshark}, reducing the communication, storage, and computational complexities of Bullshark---a state-of-the-art DAG-based consensus protocol---by a factor of $\frac{n}{\lambda}$ (where $\lambda$ is the security parameter).

The surprising aspect of our result is that, unlike other solutions that utilize sampling for scaling Byzantine consensus~\cite{gilad2017algorand,coinicidence,abraham2025asynchronous,blum2020asynchronous}, we maintain the optimal resilience of $f<n/3$.

We demonstrate the advantages of Sparse Bullshark compared to the baseline in terms of both throughput and latency on systems with bounded network bandwidth in empirical simulations.

We leave the exploration of \emph{sparse} variants of other DAG-based protocols, such as~\cite{cordialminers,mysticeti-tr,mahimahi,shoal,shoalplusplus}, for future work.

\bibliographystyle{ACM-Reference-Format}
\bibliography{references}

\appendix

\section{Correctness proofs}
\label{app:proofs}

We say that the local DAG view of user $p_i$, denoted $\dagvar_i$, is \emph{menacing} if there exists a round $r$ with an anchor $v$ and a round $r' > r$ with an anchor $v'$ such that:
\begin{enumerate}
  \item $v$ and $v'$ are both present in $\dagvar_i$;
  \item $v$ has been directly committed by $p_i$, i.e. $|\{u \in \dagvar_i[r + 1]~|~v \in u.\vEdges\}| \geq 2f + 1$;
  \item $v$ is not committed indirectly when committing $v'$, i.e., there is no path from $v'$ to $v$ in $\dagvar_i$.
\end{enumerate}

\begin{lemma}
\label{lemma:menace}
If no correct user ever has a menacing DAG, then the algorithm satisfies the \textbf{Total order} property.
\end{lemma}

\begin{proof}
First we show that in every run, the correct validators agree on the sequence of committed anchors.

Suppose a validator $p_i$ committed anchor $v$ in round $r$ using the direct commit rule.
As we assume that no DAG is menacing, in $\textit{DAG}_i$ for any anchor $v'$ in any round $r' > r$ there is a path from $v'$ to $v$ and, thus, when any correct validator commits $v'$, it will also commit $v$ by the direct or indirect commit rule.
Thus, if some correct validator commits an anchor via the direct commit rule, all correct validators will commit that anchor as well.
Hence, any two sequences of anchors committed by correct validators are related by containment: one of them is a prefix of the other.

Because we order causal histories in a deterministic order when committing anchors, if two sequences of anchors are equal then corresponding fully ordered DAG histories are also the same.
Thus, the \textbf{Total order} property is satisfied if no correct user ever has a menacing DAG.
\end{proof}

\begin{lemma}
\label{lem:fixed-pair-of-anchors}
Consider a directly committed anchor node $v$ from round $r$ and any anchor node $v'$ from round $r' > r$.
With high probability, $v'$ has a path to $v$.
\end{lemma}

\begin{proof}
Since $v$ has been directly committed by some correct user in round $r$, there exists a set $V_1$ of $2f + 1$ nodes in round $r + 1$ that have edges to $v$.
Let $\overline{V_1}$ be the set of nodes in round $r + 1$ that do \emph{not} reference $v$.
By the properties of reliable broadcast, each validator, whether correct or Byzantine, can produce at most one node in each round.
Hence, $|\overline{V_1}| \le n - |V_1| \le n - (2f + 1)$.
For $n = 3f + 1$, $|\overline{V_1}| \le f$.

Consider any vertex $u$ in round $r + 2$ that is reachable from $v'$ (it must exist, as $r' \geq r + 2$ and there are at least $\SampleSize \ge 1$ edges from each vertex).
We show that, with high probability, $u$ has an edge to some node in $V_1$.
This would imply that there is a path from $u$ to $v$, and, hence, from $v'$ to $v$.

We now show that with high probability at least one vertex from $V_1$ will be included in the sample of outgoing edges of $u$.

First phase of sampling performed upon creation of $u$ is the vector commitment to all the certificates received by the corresponding validator in round $r + 1$.
Despite the fact that a Byzantine validator might not actually commit to all of the received certificates, usage of this commitment as a seed for ALBA guarantees that the further sampling procedure actually depends on certificates of all potential predecessors of $u$.
Thus, the adversary can not choose which vertices will reference the anchor based on which vertices have been selected to the sample, ensuring that any randomness occuring in the sampling is independent of the actual composition of $V_1$.

If all the sampled vertices are contained in $\overline{V_1}$, then $u$ contains a valid ALBA proof of sampling from a set of size less than or equal to $f$.
By the properties of ALBA, the probability of the adversary constructing such a proof is negligible.
Thus, with high probability the adversary is not able to equivocate and avoid sampling any vertices from $V_1$.
\end{proof}

\begin{lemma}
\label{lem:all-pairs-of-anchors}
The probability of any correct validator having a menacing $\dagvar_i$ in an execution of the protocol of polynomial length $R$ is negligible.
\end{lemma}

\begin{proof}
Follows from the union bound and \Cref{lem:fixed-pair-of-anchors} since a polynomial-length execution will have at most a polynomial number of pairs of anchors.
\end{proof}

\ThTO*

\begin{proof}
Follows directly from \Cref{lemma:menace,lem:all-pairs-of-anchors}
\end{proof}

\begin{lemma}
\label{lemma:anchor-reference}
If the value of the timeout is at least $2 \Delta$, after the GST all anchors created by correct validators will be committed.
\end{lemma}

\begin{proof}
We rely on the following property of reliable broadcast: if an honest party delivers a message at time $t$ then all honest parties deliver this message by time $t + \Delta$.
This property is guaranteed by any reliable broadcast protocol that echoes a certificate proving the message delivery before actually delivering it.

Suppose a validator $p_i$ broadcasts an anchor in round $r$.
Let $t$ be the moment of time when the first correct validator $p_j$ advanced to round $r$ and created its vertex.
It means that at time $t$ validator $p_j$ has collected a quorum of vertices from round $r - 1$.
Due to the property of reliable broadcast, by the time $t + \Delta$ all correct validators will collect the same quorum and advance to the round $r$.
That implies that all the validators including $p_i$ will advance to round $r$ withing the interval of time $[t, t + \Delta]$.
Thus, using the property of reliable broadcast once again, we conclude that all the vertices created in round $r$ will be delivered to all the correct validators withing the time interval of length $2 \Delta$.
If the timeout used in the DAG construction is at least $2 \Delta$, the anchor will be referenced by all the correct validators and thus committed using the direct rule.
\end{proof}

\ThVA*

\begin{proof}
If a message $m$ is delivered by a correct validator, it means that an anchor $v$ containing the corresponding DAG vertex in its causal history has been committed.
By \Cref{lemma:anchor-reference} after the GST all the anchors will be committed.
Because with high probability no correct validator ever has a menacing DAG (\Cref{lem:all-pairs-of-anchors}), commiting any anchor implies commiting all the anchors from previous rounds that have been directly committed by at least one correct validator.
Thus, after the GST all the correct validators will commit $v$ and thus deliver $m$.
\end{proof}

\end{document}

%% file: pseudocode/definitions.tex
\NewDeclarationType{Variable}{\textit}
\NewDeclarationType{Function}{\textrm}
\NewDeclarationType{Event}{\textit}

\DeclareGlobalVariable[vRound]{round}
\DeclareGlobalVariable[vSource]{source}
\DeclareGlobalVariable[vBlock]{block}
\DeclareGlobalVariable[vEdges]{edges}
\DeclareGlobalVariable[vSignedRound]{signedRound}
\DeclareGlobalVariable[vProof]{sampleProof}

\DeclareGlobalVariable[dagvar]{DAG}
\DeclareGlobalVariable{blocksToPropose}
\DeclareGlobalVariable[currentRound]{r}
\DeclareGlobalVariable{buffer}
\DeclareGlobalVariable{orderedVertices}
\DeclareGlobalVariable{lastOrderedRound}
\DeclareGlobalVariable{orderedAnchorsStack}
\DeclareGlobalFunction[PathExists]{path}
\DeclareGlobalFunction[GetAnchor]{get\_anchor}
\DeclareGlobalFunction[CreateNewVertex]{create\_new\_vertex}
\DeclareGlobalFunction[ValidateVertex]{validate\_vertex}
\DeclareGlobalFunction[DirectCommitThreshold]{direct\_commit\_threshold}
\DeclareGlobalFunction[MayAdvanceRound]{may\_advance\_round}

\DeclareGlobalFunction[TryCommitting]{try\_commiting}
\DeclareGlobalFunction[OrderAnchors]{order\_anchors}
\DeclareGlobalFunction[OrderHistory]{order\_history}

\DeclareGlobalFunction[Sign]{sign}
\DeclareGlobalFunction[Aggregate]{aggregate}
\DeclareGlobalFunction[ValidateSig]{validate\_signature}
\DeclareGlobalVariable[SampleSize]{D}
\DeclareGlobalFunction[RandomChoice]{random\_sample}
\DeclareGlobalFunction[Hash]{hash}
\DeclareGlobalFunction[VerifiablySample]{verifiable\_sample}
\DeclareGlobalFunction[ValidateSample]{validate\_sample}

\DeclareGlobalVariable*[lastEpoch]{e_{\VariableStyle{last}}}

\newcommand\StateX{\Statex\hspace{\algorithmicindent}}
\newcommand\StateXX{\StateX\hspace{\algorithmicindent}}
\algrenewcommand\textproc{}% Used to be \textsc

\algdef{SE}[Upon]{Upon}{EndUpon}[1]{\textbf{upon}\ #1\ \algorithmicdo}{\algorithmicend\ \textbf{}}%
\algtext*{EndUpon}

%% file: pseudocode/common.tex
\begin{algorithm}
  \caption{Data structures and basic utilities for process $p_i$}
  \label{util}
  \begin{smartalgorithmic}[1]
    \Procedure{\PathExists}{$v$, $u$}
      \Comment{Check if exists a path from $v$ to $u$ in the $\dagvar$}
      \State \Return{$\exists [v_1, v_2, \ldots, v_k]: v_1 = v \land v_k = u \land \forall j \in [2..k]: v_j \in \bigcup_{r \geq 1} \dagvar_i[r] \land v_j \in v_{j - 1}.\vEdges$}
    \EndProcedure

    \Statex

    \Procedure{\GetAnchor}{$r$}
      \DeclareLocalVariable{p}
      \If{$r \bmod 2 = 1$}
        \State \Return{$\bot$}
      \EndIf
      \State $\p \gets (r/2) \bmod n$
      \Comment{Choose leader in round robin fashion}
      \If{$\exists v \in \dagvar_i[r]: v.\vSource = \p$}
        \State \Return{$v$}
      \EndIf
      \State \Return{$\bot$}
    \EndProcedure
  \end{smartalgorithmic}
\end{algorithm}

\begin{algorithm}
  \caption{DAG construction for process $p_i$}
  \label{construction}
  \begin{smartalgorithmic}[1]
    \Statex \textbf{Local variables:}
      \StateX $\buffer \gets \{\}$
      \StateX $\currentRound \gets 0$
    \Statex

    \Procedure{\MayAdvanceRound}{}
      \If{$|\dagvar_i[\currentRound]| < 2f + 1$}
        \LineComment{Quorum of nodes is neccessary to advance}
        \State \Return{False}
      \EndIf
      \If{\textit{timer expired}}
        \State \Return{True}
      \EndIf
      \If{$\currentRound \bmod 2 = 0$}
        \State \Return{$\GetAnchor(\currentRound) \in \dagvar_i[\currentRound]$}
      \Else
        \DeclareLocalVariable[prevAnchor]{A}
        \State $\prevAnchor \gets \GetAnchor(\currentRound - 1)$
        \State \Return{$|\{v \in \dagvar_i[\currentRound]~|~\prevAnchor \in v.\vEdges\}| \geq 2f + 1 \lor |\{v \in \dagvar_i[\currentRound]~|~\prevAnchor \notin v.\vEdges\}| \geq f + 1$}
      \EndIf
    \EndProcedure

    \Statex

    \While{True}
      \For{$v \in \buffer$}
        \If{$\forall v' \in v.\vEdges: v' \in \bigcup_{\currentRound \geq 1} \dagvar_i[r]$}
          \State $\dagvar_i[v.\vRound] \gets \dagvar_i[v.\vRound] \cup \{v\}$
          \State $\buffer \gets \buffer \setminus \{v\}$
          \State $\TryCommitting(v)$
        \EndIf
      \EndFor

      \If{\MayAdvanceRound}
        \State $\currentRound \gets \currentRound + 1$
        \State $v \gets \CreateNewVertex(\currentRound)$
        \State $r\_bcast_i(v, \currentRound)$
        \State \textit{reset timer}
      \EndIf
    \EndWhile

    \Statex

    \Upon{$r\_deliver(v, round, p_k)$}
      \If{$\ValidateVertex(v, round, k)$}
        \State $\buffer \gets \buffer \cup \{v\}$
      \EndIf
    \EndUpon
  \end{smartalgorithmic}
\end{algorithm}

\begin{algorithm}
  \caption{Ordering for process $p_i$}
  \label{ordering}
  \begin{smartalgorithmic}[1]
    \DeclareLocalVariable{anchor}
    \DeclareLocalVariable{votes}
    \DeclareLocalVariable[varR]{r}
    \DeclareLocalVariable{prevAnchor}
    \DeclareLocalVariable{verticesToOrder}

    \Statex \textbf{Local variables:}
      \StateX $\orderedVertices \gets \{\}$
      \StateX $lastOrderedRound \gets 0$
      \StateX $\orderedAnchorsStack \gets \text{empty stack}$
    \Statex

    \Procedure{\TryCommitting}{$v$}
      % \If{$v.\vRound \bmod 2 = 1$  or $v.\vRound = 0$}
      %   \State \Return{}
      % \EndIf
      \State $\anchor \gets \GetAnchor(v.\vRound - 1)$
      \If{$\anchor = \bot$}
        \State \Return{}
      \EndIf
      \If{$|\{v' \in \dagvar_i[v.\vRound]: \anchor \in v'.\vEdges\}| \geq \DirectCommitThreshold()$}
        \State $\OrderAnchors(\anchor)$
      \EndIf
    \EndProcedure

    \Statex

    \Procedure{\OrderAnchors}{$v$}
      \If{$v.\vRound \leq \lastOrderedRound$}
        \State \Return{}
      \EndIf
      \State $\anchor \gets v$
      \State $\orderedAnchorsStack.\textit{push}(\anchor)$
      \State $\varR \gets \anchor.\vRound - 2$
      \While{$\varR > \lastOrderedRound$}
        \State $\prevAnchor \gets \GetAnchor(\varR)$
        \If{$\PathExists(\anchor, \prevAnchor)$}
          \State $\orderedAnchorsStack.\textit{push}(\prevAnchor)$
          \State $\anchor \gets \prevAnchor$
        \EndIf
        \State $\varR \gets \varR - 2$
      \EndWhile
      \State $\lastOrderedRound \gets v.\vRound$
      \State $\OrderHistory()$
    \EndProcedure

    \Statex

    \Procedure{\OrderHistory}{}
      \While{$\neg \orderedAnchorsStack.\textit{isEmpty}()$}
        \State $\anchor \gets \orderedAnchorsStack.\textit{pop}()$
        \State $\verticesToOrder \gets \{v \in \bigcup_{r > 0} \dagvar_i[r]
            \mid \PathExists(\anchor, v) \land v \not\in \orderedVertices\}$
          \For{$v \in \verticesToOrder$ in some deterministic order}
            \State \textbf{order} $v$
            \State $\orderedVertices \gets \orderedVertices \cup \{v\}$
          \EndFor
      \EndWhile
    \EndProcedure
  \end{smartalgorithmic}
\end{algorithm}

%% file: pseudocode/es_bullshark.tex
\begin{algorithm}
  \caption{Eventually synchronous bullshark: protocol-specific utilities}
  \label{esbs}
  \begin{smartalgorithmic}[1]
    \Statex \textbf{Local variables:}
      \StateX struct $\textit{vertex } v$: \Comment{The struct of a vertex in the DAG}
        \StateXX $v.\vRound$ - the round of $v$ in the DAG
        \StateXX $v.\vSource$ - the process that broadcast $v$
        \StateXX $v.\vBlock$ - a block of transactions information
        \StateXX $v.\vEdges$ - a set of $n - f$ vertices in $v.\vRound - 1$
      \StateX $\dagvar_i[]$ - an array of sets of vertices

    \Statex

    \Procedure{\CreateNewVertex}{$round$}
      \State \textbf{wait until} $\neg \blocksToPropose.\text{empty}$()
      \State $v.\vBlock \gets \blocksToPropose.\text{dequeue}()$
      \State $v.\vRound \gets round$
      \State $v.\vSource \gets i$
      \State $v.\vEdges \gets \dagvar_i[round - 1]$
      \State \Return{$v$}
    \EndProcedure

    \Statex

    \Procedure{\ValidateVertex}{$v$, $round$, $k$}
      \State \Return{$v.\vSource = k \land v.\vRound = round \land |v.\vEdges| \geq 2f + 1$}
    \EndProcedure

    \Statex

    \Procedure{\DirectCommitThreshold}{}
      \State \Return{f + 1}
    \EndProcedure
  \end{smartalgorithmic}
\end{algorithm}

%% file: pseudocode/sparse_bs.tex
\begin{algorithm}
  \caption{Sparse Bullshark: protocol-specific utilities}
    \label{alg:new}
  \newcommand{\diff}[1]{\textcolor{blue}{#1}}

  \begin{smartalgorithmic}[1]
    \Statex \textbf{Local variables:}
      \StateX struct $\textit{vertex } v$: \Comment{The struct of a vertex in the DAG}
        \StateXX $v.\vRound$ - the round of $v$ in the DAG
        \StateXX $v.\vSource$ - the process that broadcast $v$
        \StateXX $v.\vBlock$ - a block of transactions information
        \StateXX $v.\vEdges$ - a set of \diff{$\SampleSize$} vertices in $v.\vRound - 1$
        % \StateXX \diff{$v.\vSignedRound$ - signed value of $v.\vRound$}
        \StateXX \diff{$v.\vProof$ - proof of sampling from a set of $2f + 1$ vertices from the previous round}

      \StateX $\dagvar_i[]$ - an array of sets of vertices

    \Statex

    {\color{blue}
    \Procedure{\VerifiablySample}{$S$}
      \DeclareLocalVariable{vectorCommitment}
      \DeclareLocalVariable{sample}
      \DeclareLocalVariable{openingProofs}
      \State $\vectorCommitment = VC.Commit(S)$
      \State $\sample = ALBA.Prove(\Hash(\vectorCommitment), S)$
      \State $\openingProofs = \{VC.Proof(S, i)~|~i \in sample\}$
      \State \Return $\sample, (\vectorCommitment, \openingProofs)$
    \EndProcedure
    }

    \Statex

    {\color{blue}
    \Procedure{\ValidateSample}{$sample$, $proof$}
      \DeclareLocalVariable{vectorCommitment}
      \DeclareLocalVariable{sample}
      \DeclareLocalVariable{openingProofs}
      \State $(\vectorCommitment, \openingProofs) \gets proof$
      \For{$v \in \sample$}
        \If{$\neg VC.Verify(\vectorCommitment, v.\vSource, v, \openingProofs[v])$}
          \State \Return $false$
        \EndIf
      \EndFor
      \State \Return $ALBA.Verify(\Hash(\vectorCommitment), sample)$
    \EndProcedure
    }

    \Statex

    \Procedure{\CreateNewVertex}{$round$}
      \State \textbf{wait until} $\neg \blocksToPropose.\text{empty}$()
      \State $v.\vBlock \gets \blocksToPropose.\text{dequeue}()$
      \State $v.\vRound \gets round$
      \State $v.\vSource \gets i$
      \DeclareLocalVariable{sample}
      \DeclareLocalVariable{samplingProof}
      {\color{blue}
      \State $\sample, \samplingProof = \VerifiablySample(\dagvar_i[round - 1], D)$
      \State $v.\vProof \gets \samplingProof$
      \State $v.\vEdges \gets \sample \cup \{\dagvar_i[round - 1][i]\}$
      }
      
      \If{$\GetAnchor(round - 1) \neq \bot$}
        \Comment {Always try to link to the wave leader}
        \State $v.\vEdges \gets v.\vEdges \cup \{\GetAnchor(round - 1)\}$ \label{line:include-anchor}
      \EndIf
      \State \Return $v$
    \EndProcedure

    \Statex

    \Procedure{\ValidateVertex}{$v$, $round$, $k$}
      \If{$v.\vSource \neq k \lor v.\vRound \neq round \diff{~\lor~|v.\vEdges| > \SampleSize + 2}$}
        \Statex \Comment{There can be at most $\SampleSize + 2$ edges from a vertex: $\SampleSize$ selected, source and anchor}
        \State \Return{False}
      \EndIf
      {\color{blue}
      \If{$\neg \ValidateSample(v.\vEdges, v.\vProof)$}
        \State \Return{False}
      \EndIf
      }
      \State \Return{True}
    \EndProcedure

    \Statex

    \Procedure{\DirectCommitThreshold}{}
      \State {\color{blue} \Return{2f + 1}}
    \EndProcedure
  \end{smartalgorithmic}
\end{algorithm}